\documentclass[12pt,a4paper]{elsarticle}

\usepackage{lineno,hyperref}
\modulolinenumbers[5]

\journal{XXX}

\usepackage{amssymb}
\usepackage{amsmath}
\usepackage{amsfonts}
\usepackage{amsthm}
\usepackage{wasysym}

\usepackage[ansinew]{inputenc}
\usepackage[T1]{fontenc}
\usepackage{lmodern}
\usepackage{textcomp}
\usepackage{graphicx}

\usepackage[]{color}

\usepackage{fancyhdr}
\usepackage{algorithm}

\usepackage{algpseudocode}

\usepackage{tcolorbox}
\tcbuselibrary{breakable}
\usepackage{caption}

\newtheorem{lemma}{Lemma}

\newtheorem{corollary}{Corollary}
\newtheorem{prop}{Proposition}

\theoremstyle{definition}

\newtheorem{definition}{Definition}

\newtcolorbox{algbox}[3][]{
  width=\textwidth,colframe=black,colback=white,
  sharp corners,
  before={\captionof{algorithm}{#2}\label{#3}},
  #1
}

\makeatletter
\renewcommand{\ALG@name}{Procedure}
\makeatother

\begin{document}

\begin{frontmatter}

\title{The Evolving Moran Genealogy}


\author{Johannes Wirtz\corref{mycorrespondingauthor}}
\cortext[mycorrespondingauthor]{Corresponding author}
\ead{wirtzj0@uni-koeln.de}

\author{Thomas Wiehe\corref{}}
\ead{twiehe@uni-koeln.de}

\begin{abstract}
We study the evolution of the population genealogy in the classic neutral Moran Model of finite size $n\in\mathbb{N}$ and in discrete time. The stochastic transformations that shape a Moran population can be realized directly on its genealogy and give rise to a process on a state space consisting of $n$-sized binary increasing trees. We derive a number of properties of this process, and show that they are in agreement with existing results on the infinite-population limit of the Moran Model. Most importantly, this process admits time reversal, which makes it possible to simplify the mechanisms determining state changes, and allows for a thorough investigation of the Most Recent Common Ancestor process.
\end{abstract}


\end{frontmatter}
\section{Introduction}
The Moran Model \cite{moran:model} is a fundamental population model of evolutionary biology. It has been used to study the evolution of a population of fixed size containing individuals of differing allelic types, subject to neutral or selective drift. In the large-population limit, its dual process describing population and sample genealogies is Kingman's Coalescent \cite{kingman:coalescent}. A very intuitive graphical representation of the infinite-population Moran Model is the Lookdown-Construction \cite{donnelly:lookdown,donnelly:lookdowna,donnelly:lookdownb}, which has also proven useful in the analysis of genealogical traits of a Moran population; for instance, it can be used to study the underlying process of \textit{Most Recent Common Ancestors} (\textit{MRCA}'s), i.e., the speed of evolution, or loss of information on the past \cite{pfaffelhuber:mrca}. A related approach to study the genealogy of a Moran population is to interpret the genealogy as a \textit{metric measure space}, which leads to a measure-valued Fleming-Viot process in the infinite-population limit \cite{depperschmidt:fleming}.\\
It is well known that genealogies of finite samples from such an infinite population can be represented in a discrete setting making use of the Yule process \cite{steel:yule}. This process generates random trees in the graph-theoretical sense and is often interpreted as a model of speciation. The distribution on trees of a given size that it induces, however, is equivalent to Kingman's Coalescent of finite size with respect to graph-theoretical and statistical properties of the trees it generates \cite{aldous:probability}. The aforementioned duality in turn establishes a connection between Moran and Yule processes.\\
In this work, we observe the evolution of the genealogy of a finite Moran population. We call this process the \textit{Evolving Moran Genealogy} (\textit{EMG}). It turns out that the state space of the \textit{EMG} can also be represented by the Yule process of finite iterations. In order to provide a detailed description of the \textit{EMG} as a Markov chain we discuss explicitly the Yule process and its associated the tree structures. We make use of this construction to observe genealogical properties of an evolving Moran population. This gives rise to the finite-population analogue of the limiting tree balance process described before \cite{pfaffelhuber:mrca,delmas:families}. This is related to the so-called root-jump process, also referred to as the \textit{MRCA}-process \cite{pfaffelhuber:mrca}. Again, here we study its 
finite-population counterpart.
A crucial feature of this discrete setting is the time-reversibility of the \textit{EMG}. The time-reversed process, denoted by $\textit{EMG}^\flat$, is algorithmically simpler than the \textit{EMG},
because it requires only grafting of branches, instead of the two independent processes of splitting and killing. Therefore, the analysis of genealogical properties over time, e.g. of the \textit{MRCA}- process, is simplified. 
Furthermore, the consideration of the reversed process offers new and interesting insights on the "age structure" and persistence time of tree nodes.

\section{Material}
\subsection{Trees Generated under the Yule Process}\label{subsec:yp}
Many variations of the original Yule Process, or Yule Model, as defined by G. U. Yule \cite{yule:process}, have been considered throughout the literature of mathematical population genetics. One very basic, discrete version of the process is described in Procedure~\ref{proc:yule} (see also \cite{steel:yule}):\\


\vbox{
\vspace{5pt}
\begin{algbox}{Discrete Yule Process}{proc:yule}
\begin{algorithmic}[1]
\State Start with a tree consisting of one single leaf node $\iota$.
\While{Tree has $k<n$ leaves}
\State Choose one leaf $\iota$ uniformly, label it by the current total number of leaves, turn it into an internal node $\nu$ with label $k$ and append two new leaves to it.
\EndWhile
\Ensure Tree with $n$ leaves
\end{algorithmic}
\end{algbox}
}
\vbox{
\begin{definition}
\vskip 0pt ~
\begin{itemize}
\item A tree $T$ generated according to Procedure~\ref{proc:yule} is called a (random) \textit{Yule tree}.
\item The \textit{size} $|T|$ of $T$ is given by the number $n$ of leaves.
\item Let $\iota = T^{(1)},\cdots,T^{(n)}=T$ denote the trees at intermediate iterations.
\end{itemize}
\end{definition}
}
In the context of \cite{ford:alpha}, Procedure~\ref{proc:yule} corresponds to the $\alpha$-model with $\alpha=0$. We ssume that appending of leafs is graphically carried out in downward direction and in such a way that $T$ is a plane graph. In particular, this means that in each iteration, one of the leaves is appended to the bottom left and one to the bottom right. 
This induces an orientation on the branch pairs appended below an internal node, and a horizontal ordering of the leaves of $T$, which allows us to denote them by $\iota_1,\dots,\iota_n$ from left to right. Similarly, identifying the index of an internal node with its label assigned by the Yule process, we may denote the internal nodes consecutively by $\nu_1,\dots,\nu_{n-1}$. If $n\geq 2$, the internal node $\nu_1$ is of (total) degree $2$ and is called \textit{root} of $T$, while all other internal nodes are of degree $3$. $T$ has exactly $2n-2$ edges (\textit{branches}), and for any leaf $\iota\in\{\iota_1,\dots,\iota_n\}$, when moving downward on the unique path from $\nu_1$ toward $\iota$, the sequence of integer labels of internal nodes on this path is increasing. Because of that, $T$ can be interpreted as an object of the class of \textit{binary increasing trees} \cite{donaghey:bij} that are additionally equipped with an orientation on the branches, i.e., ordered (e.g. \cite{flajolet:combinatorics}, pp. 143-144).
\vbox{
\begin{definition}
\vskip 0pt ~
\begin{itemize}
\item An \textit{ordered binary increasing tree} $T$ of size $n$ is a rooted binary tree with $n$ leaves and $n-1$ internal nodes carrying unique labels from the set $\{1,\dots,n-1\}$, such that the sequence of integers encountered on any path from the root to a leaf is increasing, and for any internal node, it holds that one of the two branches attached below (from the root) points to the left, and the other one to the right. 
\item Let $\mathcal{T}_n$ denote the set of all binary increasing trees of size $n$.
\end{itemize}
\end{definition}
}
The size of $\mathcal{T}_n$ equals $(n-1)!$, and it is easily deduced that for any $T\in\mathcal{T}_n$, there exists precisely one sequence of possible iterations of the Yule process such that the resulting Yule tree equals $T$. Furthermore, since leaves are always chosen uniformly in the Yule process, the probability of obtaining some $T\in\mathcal{T}_n$ under the Yule Process is $\Pr(T)=\frac{1}{(n-1)!}$, i.e., the Yule tree is uniformly distributed on $\mathcal{T}_n$.\\
\begin{figure}
\includegraphics[scale=.5]{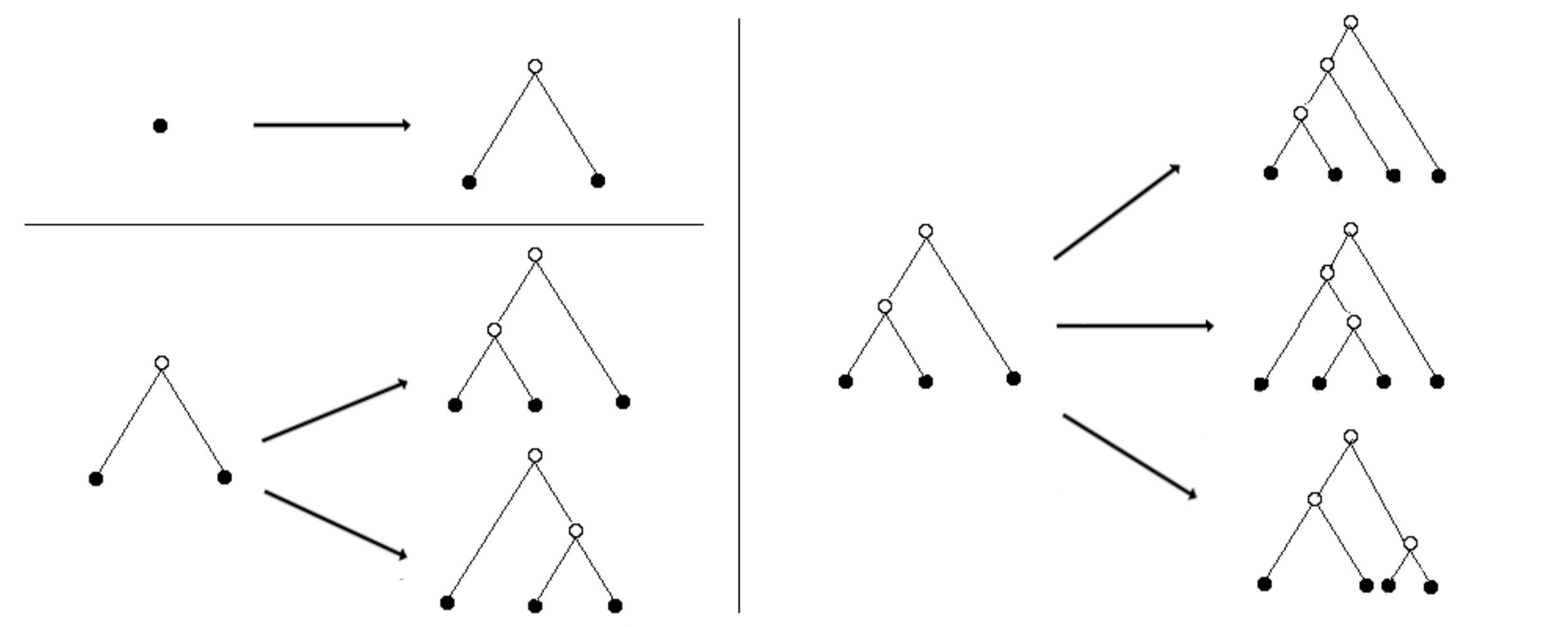}
\caption{Some possible iterations of the Yule tree-generating procedure}
\label{fig:yp}
\end{figure}
Suppose further that all $n$ leaves of $T$ are drawn on the same vertical "height" $0$, and all internal nodes $\nu_k$ on height $n-k$. Then, $T$ divides the plane into $n$ layers $l_1,\dots,l_n$, where layer $l_k,k=2,\dots,n-1$ is vertically restricted by the heights of $\nu_{k-1}$ and $\nu_{k}$. Layer $1$ extends upwards to infinity from the root's height, and layer $n$ from height $1$ to $0$. If $k\geq 2$, the $k$'th layer of $T$ is the layer which is crossed by precisely $k$ branches. This notion can be extend to layer $1$ by assuming that it contains an \textit{imaginary branch} extending from the root upwards. We may think of a branch $\beta$ as a composite of \textit{branch segments}, where a segment only extends over one layer. Then $T$ contains $1+2+\dots+n=\frac{n(n+1)}{2}$ such segments (counting the imaginary branch as a single segment). We denote them by $b_1,\dots,b_{\frac{n(n+1)}{2}}$ from top to bottom and left to right (see Figure \ref{fig:yt2}).\\
\begin{figure}
\includegraphics[scale=.5]{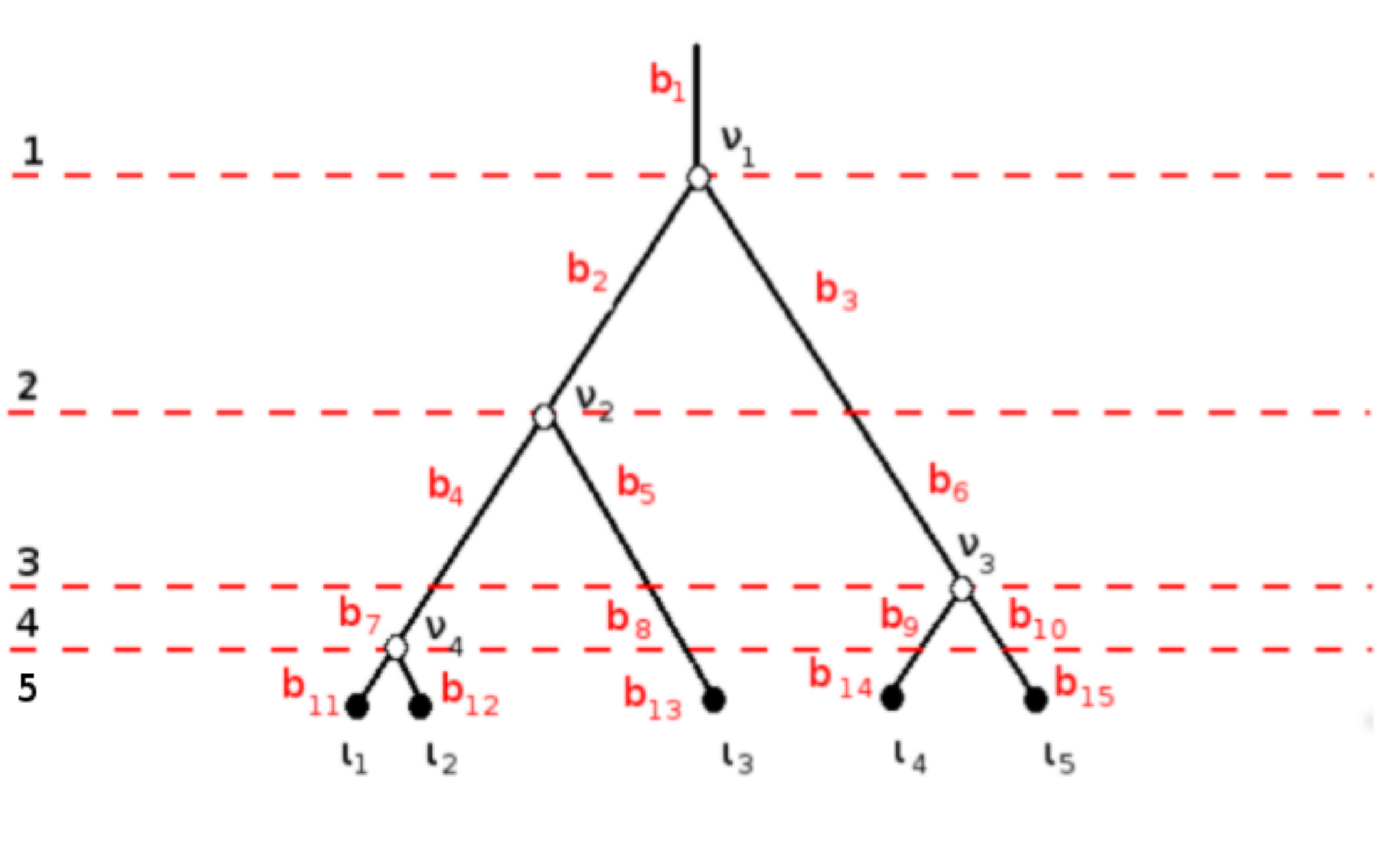}
\caption{A Yule tree of size $5$ with all layer, branch segment and node labellings depicted}
\label{fig:yt2}
\end{figure}
To simplify the following calculations, we define the following notation:\\
\vbox{
\begin{enumerate}
\item  We denote by $\sigma_T(i)$ the mapping from an integer $i,i=1,\dots,n-1$ to the leaf $\iota$ of $T^{(i)}$ chosen in the $i$'th iteration of the process generating $T$.
\item  We denote by $l(b)$ the layer across which a segment $b$ extends.
\end{enumerate}
}
Let $S$ denote a set of leaves of some $T\in\mathcal{T}_n$. Connecting all leaves of $S$ according to the branching pattern of $T$ generates another tree $T_{S}$ on $|S|$ leaves, where $|S|-1$ internal nodes of $T$ are preserved. If we label the internal nodes of $T_{S}$ by $1,\dots,|S|-1$ such that their relations with respect to height are preserved from $T$, $T_{S}\in\mathcal{T}_{|S|}$. Each leaf $\iota'$ in $T_S$ equals some leaf $\iota\in \{\iota_1,\dots,\iota_n\}$ of $T$, and the horizontal order of leaves in $T_S$ is in accordance with that in $T$. Similarly, each internal node $\nu'_k$ in $T_S$ is representative of some internal node $\nu_{l}$ in $T$, with $k\leq l$.\\
\vbox{
\begin{definition}
For $T\in\mathcal{T}_n$ and $\emptyset\neq S\subseteq\{\iota_1,\dots,\iota_n\}$:
\begin{itemize}
\item The object $T_{S}$ is called the ($S$-)\textit{induced subtree} of $T$.
\item For an internal node $\nu'\in\{\nu'_1,\dots,\nu'_{|S|-1}\}$ of $T_S$, let $\phi(\nu')$ denote the internal node of $T$ that is represented by $\nu'$ in $T_S$.
\item For all $j=1,\dots,|S|-1$, let $\tau(j)\in\{1,\dots,n-1\}$ denote the label of $\phi(\nu'_j)$ in $T$
\end{itemize}
\end{definition}
}
\begin{figure}
\includegraphics[scale=.625]{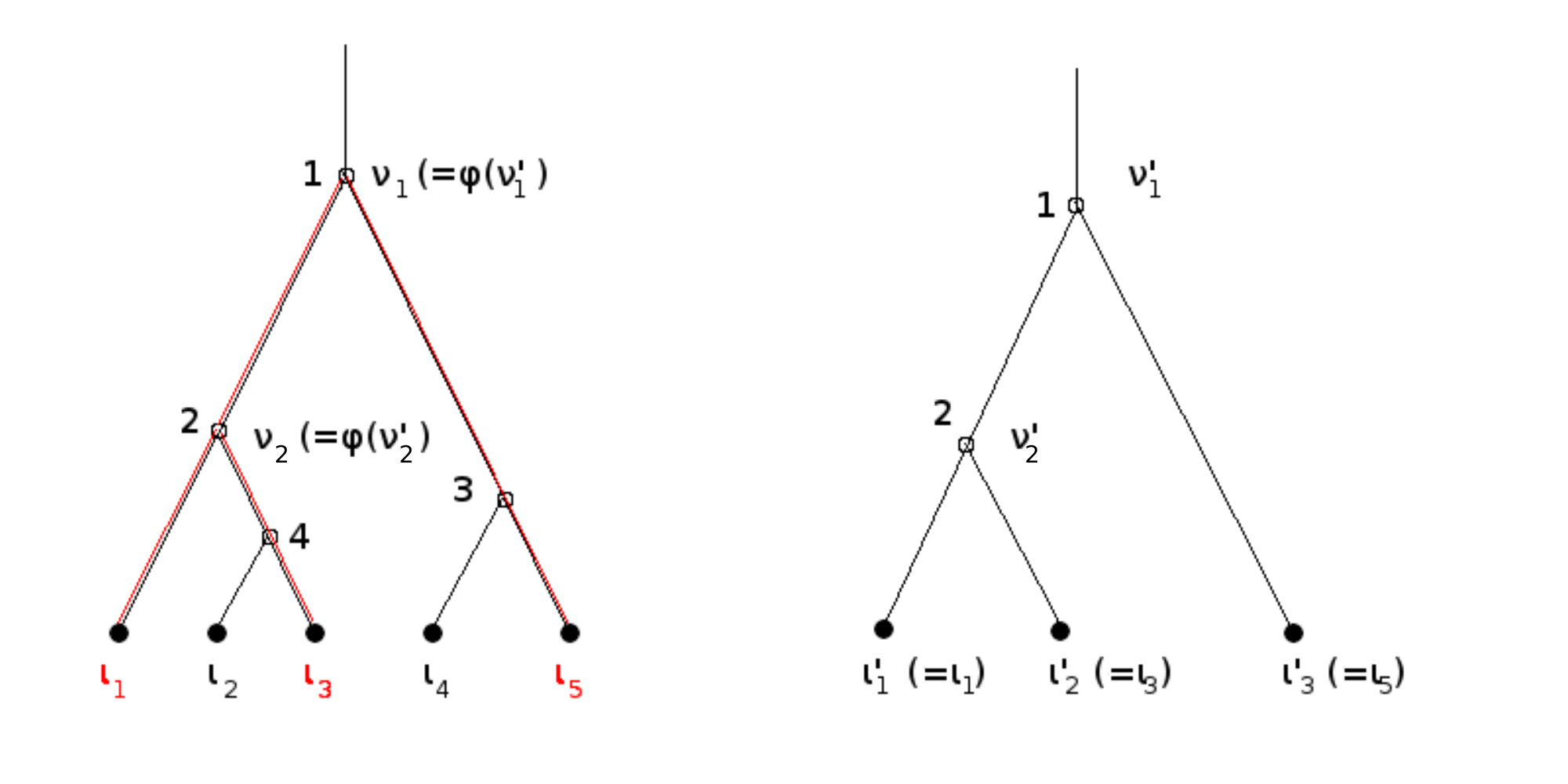}
\caption{A Yule tree of size 5 and the induced subtree of leaves $\iota_1,\iota_3,\iota_5$.}
\label{fig:induced}
\end{figure}

See Figure~\ref{fig:induced} for an example. If $S=\{\iota\}$ for some single leaf $\iota$ of $T$, $T_S$ equals the tree of size $1$ consisting just of $\iota$, and $T_{\{\iota_1,\dots,\iota_{n}\}}=T$.\\
In the following Lemma, we state that Yule trees exhibit a form of self-similarity with respect to induced subtrees.
\begin{lemma}[Sample-Subtree Invariance of Yule trees]\label{lemma:ssiyt}
Let $T$ denote a Yule tree of size $n$, and $S\subseteq \{\iota_1,\dots,\iota_n\}$, $|S|=k$, where the leaves $\iota\in S$ are chosen uniformly and without replacement. Then 
\begin{equation}
\forall \tilde{T}\in\mathcal{T}_k:\ \Pr(T_S=\tilde{T})=\frac{1}{(k-1)!}
\end{equation}
\end{lemma}
\begin{proof}
We show that we can treat $T_S$ as a tree generated by the Yule Process. Since this is obviously true for $|S|=1$ (or $S=2$), we apply induction on $k$.\\
Let $S=\{\iota'_1,\dots,\iota'_{k}\}$. Tracing back the iterations $l = n,\dots,\tau(|S|-1)$ of the process generating $T$, for each $\iota'_{j}\in S$ there is a unique leaf $\iota^{(l)}_{j}$ of $T^{(l)}$ such that either $\iota'_j=\iota^{(l)}_j$ or $\iota'_j$ is appended below $\iota^{(l)}_j$ by one or more Yule iterations. In $T^{(\tau(|S|-1)-1)}$, a leaf $\iota^{*}=\sigma_T(\tau(|S|-1))$ is turned into $\phi(\nu_{|S|-1})$ in iteration $\tau(|S|-1)$ and two of the leafs $\iota^{(\tau(|S|-1))}_m,\iota^{(\tau(|S|-1))}_{m+1}$ that are the correspondents of $\iota'_m,\iota'_{m+1}$ in $T^{(\tau(|S|-1))}$ are appended below.\\
Consider the set $S'=\{\iota^{\tau(|S|-1)}_{1},\dots,\iota^{\tau(|S|-1)}_{m-1},\iota^*,\iota^{\tau(|S|-1)}_{m+2},\dots,\iota^{\tau(|S|-1)}_k\}$. Because of the established correspondence of internal nodes between $T_S$ and $T^{(\tau(|S|-1)-1)}_{S'}$, $T_S$ is created out of $T^{(\tau(|S|-1)-1)}_{S'}$ by turning $\iota^*$ into an internal node and appending two new leaves. If $\iota^*$ is chosen uniformly from $S'$, then this simply corresponds to one Yule iteration. We verify this, writing $\Pr(\sigma_{T_S}(|S'|)=\iota^*)$ for the probability that $\iota =\iota^*$ for $\iota\in S'$:
\begin{align*}
\Pr(\sigma_{T_S}(|S'|)=\iota^*)&=\Pr\left(\sigma_T(\tau(|S|-1))=\iota^*|\sigma_T(\tau(|S|-1)) \in S'\right)\\
&=\frac{1/\tau(|S|-1)}{|S'|/\tau(|S|-1)}\\
&=\frac{1}{|S'|}
\end{align*}
In addition, the fact that $i^*$ is chosen uniformly from $S'$ implies that $S'$ can be treated as a set of size $k-1$ that is randomly chosen from the leaves of $T^{(\tau(|S|-1)-1)}$. By induction hypothesis, the induced subtree $T^{(\tau(|S|-1)-1)}_{S'}$ is then a Yule tree of size $k-1$, i.e. generated by $k-2$ iterations. Since the last step from $T^{(\tau(|S|-1)-1)}_{S'}$ to $T_S$ can be interpreted as a $k-1$'th iteration, we conclude that the process generating $T_S$ is a Yule Process of $|S|-1=k-1$ iterations.

\end{proof}
This property is similar to what is called \textit{Markovian self-similarity} in \cite{ford:alpha}. Another form of self-similarity that arises in the context of the Yule Process is \textit{Horton self-similarity}, which applies, for example, to the limit of Kingman's Coalescent \cite{kovchegov:similarity}.\\
Let again $T$ denote a Yule tree of size $n$. Instead of applying an iteration of the Yule process, $T$ can also be transformed into a Yule tree of size $n+1$ by \textit{random grafting} (Procedure~\ref{proc:rg}) a new branch leading to a leaf into $T$:\\
\vbox{
\vspace{5pt}
\begin{algbox}{Random Grafting Operation}{proc:rg}
\begin{algorithmic}[1]
\Require Yule tree $T$ of size $n$
\State Choose a branch segment $b$ uniformly from all $\frac{n(n+1)}{2}$ possible segments and an "orientation" $\chi\in \{\textit{left,right}\}$ uniformly \Comment{including the imaginary branch}
\State Split all branch segments $b',l(b')=l(b)$ into two separate branch segments \Comment{forming an additional layer}
\State Between the two pieces $b^{(1)},b^{(2)}$ resulting from splitting $b$, place a new internal node $\nu$ with label $l(b)$.
\State Increase the labels of all internal nodes in layers $k>l(b)$ by one;
\State At $\nu$, append a new branch $\beta$ consisting of $n-l(b)+1$ segments and ending in a new leaf $\iota$, to the left or right depending on $\chi$;
\State $\hat{T}\leftarrow T$

\Ensure Tree $\hat{T}$ with $n+1$ leaves
\end{algorithmic}

\end{algbox}
}
The orientation $\chi$ accounts for the fact that branches are implicitly oriented in the version of the Yule Process described above. Note that the position of the new leaf $\iota$ in $\hat{T}$ depends on $\chi$. A possible realization of Procedure~\ref{proc:rg} is depicted in Figure~\ref{fig:regraft}.\\
Applying Procedure~\ref{proc:rg}, we obtain an object $\hat{T}\in\mathcal{T}_{n+1}$. We write $T\uparrow\hat{T}$ if $\hat{T}$ was constructed from $T$ by random grafting. In total, there are $k(k+1)$ possibilities $(b,\chi)$ of performing a grafting in $T$ of equal probability, and unique with respect to which leaf and internal node of $\hat{T}$ they generate. However, different grafting operations on $T$ may generate the same object $\hat{T}$.

\begin{figure}
\includegraphics[scale=.5]{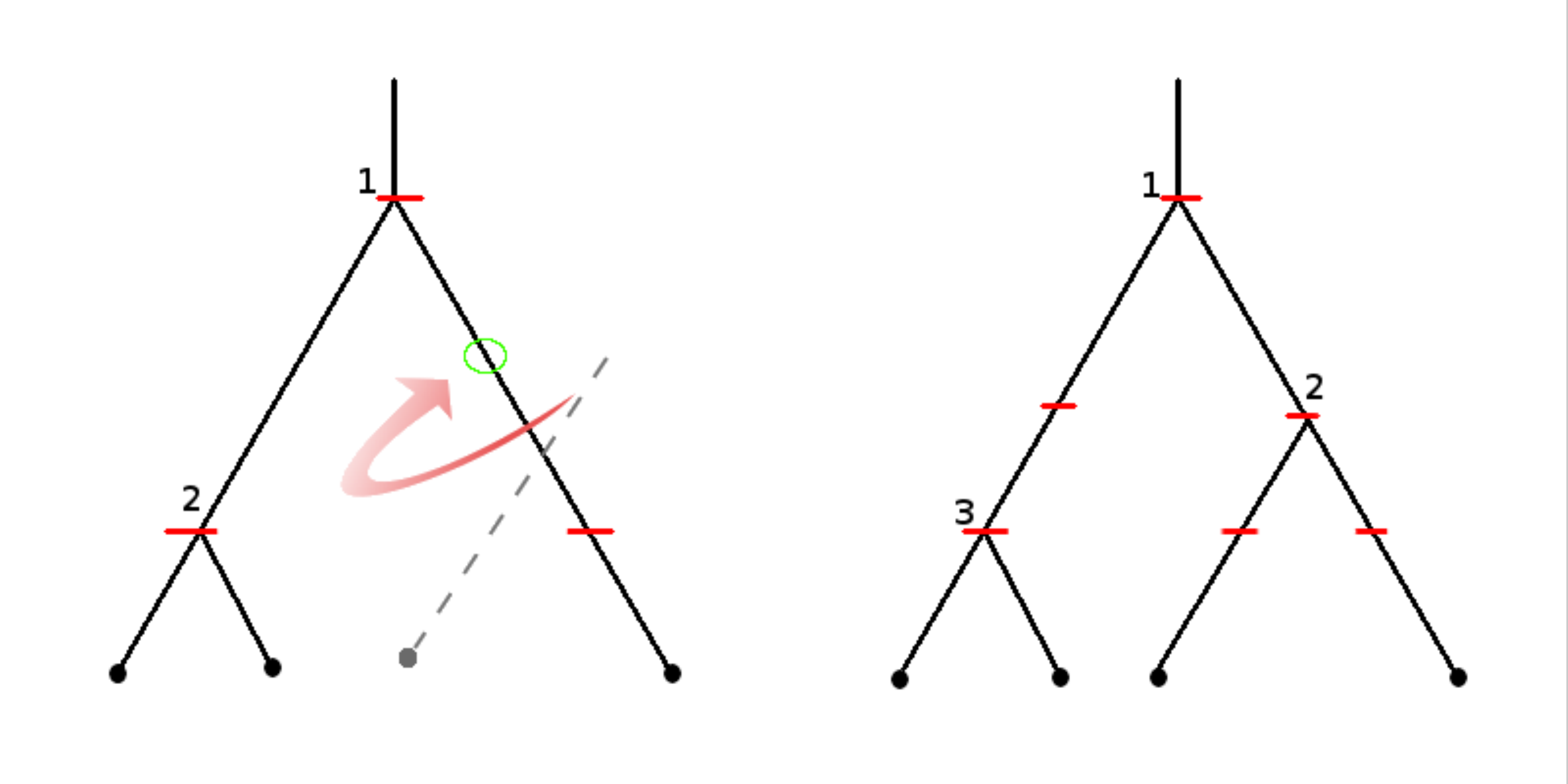}
\caption{The regrafting operation~\ref{proc:rg} performed on the branch segment with the "\textcolor{green}{o}" mark, transforming the $3$-sized tree on the left into a tree of size $4$.}
\label{fig:regraft}
\end{figure}
The relation between grafting operation and the original Yule Process is described by the following Lemma:
\begin{lemma}[Piecewise Recovery by Grafting]\label{lemma:pwr}
Let $T$ be a Yule tree of size $n$, $S=\{\iota_1',\dots,\iota_{k+1}'\}\subseteq \{\iota_1,\dots,\iota_n\}$ a set of leaves chosen uniformly without replacement, and $\iota'\in S$ chosen uniformly. Then
\begin{equation}
\forall T'\in\mathcal{T}_{k},T''\in\mathcal{T}_{k+1}:\ \Pr(T_S=T''|T_{S\setminus \iota'}=T')=\Pr(T'\uparrow T'')
\end{equation}
\end{lemma}
\begin{proof}
Let $l\in\mathbb{N}_0$ denote the number of graftings that can be performed on $T'$ to generate $T''$, thus $\Pr(T'\uparrow T'')=\frac{l}{k(k+1)}$. On the other hand,
\begin{align*}\Pr(T_S=T''|T_{S\setminus \iota'}=T')&=\frac{\Pr(T_S=T'',T_{S\setminus \iota'}=T')}{\Pr(T_{S\setminus \iota'}=T')}\\
\end{align*}
and by Lemma~\ref{lemma:ssiyt}, $\Pr(T_{S\setminus \iota'}=T')=1/(k-1)!$. Let $m\in\mathbb{N}_0$ denote the number of leafs $\iota'\in S$ such that $T_{S\setminus \iota'}=T'$. Since each tree $\tilde{T}\in\mathcal{T}_{k+1}$ is equally likely to be the induced subtree $T_S$ and $\iota'\in S$ is chosen uniformly, we have
$$\Pr(T_S=T'',T_{S\setminus \iota'}=T')=\frac{m}{k!(k+1)}$$
and thus $\Pr(T_S=T''|T_{S\setminus \iota'}=T')=\frac{m}{k(k+1)}$.\\ 
Let $\iota'\in S$ such that $T_{S\setminus \iota'}=T'$, and $\nu'$ the internal node $\iota'$ is appended to. There exists exactly one tuple $(b,\chi)$ such that, performing the associated grafting operation in $T'$, we obtain $T''$, the leaf generated by the operation occupies the position of $\iota'$ in $T''$, and the internal node generated by it carries the label of $\nu'$. Conversely, each tuple $(b,\chi)$ such that the associated grafting operation on $T'$ yields $T''$ generates a unique leaf $\iota ^*$ with respect to horizontal position and an internal node $\nu ^*$. Then, there exists a unique $\iota'\in S$ that occupies the position of $\iota ^*$ in $T_S$, and since $T''=T_S$, the induced subtree $T_{S\setminus \iota'}$ of $T_S$ equals $T'$. Therefore, $m=l$ holds, which ends the proof.
\end{proof}
We immediately conclude
\begin{corollary}
The distributions of the $n$-sized Yule tree and $T\in\mathcal{T}_n$ generated by successive random graftings are equal, therefore
$$\Pr(T|T\textnormal{ generated by random grafting})=\frac{1}{(n-1)!}$$
\end{corollary}
\begin{proof}
This follows by induction on $n$, making use of Lemma~\ref{lemma:pwr}.
\end{proof}
\subsection{The Genealogical Process of a Moran Population}
The Moran process of finite size $n\in\mathbb{N}$ in discrete time is a Markov chain whose state space consists of ordered multisets $\mathcal{P}_i=\{x_1,\dots,x_n\}$ ("populations") of objects ("individuals"), where $i\in\mathbb{N}_0$ represents time, and $\mathcal{P}_0$ is some initial set. The transition between time steps is facilitated by the following operation, incorporating two random mechanics:\\
\vbox{
\vspace{5pt}
\begin{algbox}{Iteration of the Moran process}{proc:moran}

\begin{algorithmic}[1]
\Require Population $P$ of size $N$

\State Choose one element $x_k$ of $P$ uniformly and create a copy $x_k'=x_k$
\State Choose one element $x_l$ of $P$ uniformly

\State $P'\leftarrow \{x_1,\dots,x_{l-1},x_k',x_{l+1},\dots,x_N\}$
\Ensure New population $P'$ of size $N$

\end{algorithmic}

\end{algbox}
}\\
\vbox{
\begin{definition}
The neutral Moran process is denoted by $M=(\mathcal{P}_i)_{i\in\mathbb{N}}$, where $\mathcal{P}_{i+1}$ is obtained by the application of Procedure~\ref{proc:moran} on $\mathcal{P}_i$.
\end{definition}
}
One iteration of $M$ is often interpreted as one individual of $\mathcal{P}_i$ producing offspring, and one dying. Note that $k=l$ is not excluded, therefore, there are $n^2$ possible transitions of equal probability $\frac{1}{n^2}$. Several modifications of this process exist \cite{etheridge:popmodels}; for instance, the case of an initial population consisting of two different "types" of individuals $a,A$ is known as the \textit{two-allele} Moran process. Other versions allow mutations between types or let the probabilities of producing offspring and/or dying depend on the type of individuals to model natural selection \cite{kluth:moranselection}. Here, we only consider the \textit{neutral} Moran process with uniform transition probabilities and without mutation.\\
With probability $1$, there is a finite time $i^*$ at which $M$ will enter a state in which the population consists only of the copies (descendants) of some $x_k\in \mathcal{P}_0$, while all other $x_l\in \mathcal{P}_0,l\neq k$, and their copies have been removed from the population. The individual $x_k$, or one of its descendants, is thus the \textit{Most Recent Common Ancestor} (\textit{MRCA}) dating back to at most time $0$, and looking backwards in time, there exists a branching pattern describing how the current population of $x_k$-copies has been created, in the form of a binary tree $T$ with branch lengths given implicitly by the number of time steps between splits. Considering $n\rightarrow \infty$ and rescaling time by $2/n^2$, the (infinitely large) genealogy $T$ after the first time at which there exists a \textit{MRCA}) is a realization of Kingman's Coalescent \cite{wakeley:coaltheory} and all sample genealogies are (Kingman-) Coalescent trees of sample size $n'$, $n'\in\mathbb{N}$. In this work, we focus on genealogies of finite Moran processes.\\
We assume that copy and original are indistinguishable after a duplication ($M$ is memoryless), and the copy is placed to the side of the original (instead of replacing the killed individual), and other individuals are shifted to the left or right depending on whether $l<k$ or $l>k$. This is achieved by following Procedure~\ref{proc:pmoran} and produces a "disentangled" version of the Moran process. Importantly, $T$ can then be treated as a plane graph without having to re-order individuals (see Figure~\ref{fig:plane}). Although this may seem like a minor additional complication, it is in fact crucial for the following considerations, similarly to the case of the Lookdown-Construction \cite{donnelly:lookdown,donnelly:lookdowna,donnelly:lookdownb}, whose benefits also result from intelligent organization of duplications and removals in the underlying population model.\\

\vbox{
\vspace{5pt}
\begin{algbox}{Planar order maintenance in $M$}{proc:pmoran}
\begin{algorithmic}[1]
\Require Population $P$ of size $N$

\State Choose one element $x_k$ of $P$ uniformly and create a copy $x_k'=x_k$
\State Choose one element $x_l$ of $P$ uniformly
\If{$l<k$}
\State Lower the position of individuals $x_{l+1},\dots,x_{k-1}$ by one;
\State Assign the possible positions $k-1,k$ to individual $x_k$ and its copy with probability $1/2$; 
\State $P'\leftarrow \{x_1,\dots,x_{l-1},x_{l+1},\dots,x_{k-1},x_k\textnormal{v}x_{k}',x_k\textnormal{v}x_{k}',x_{k+1},\dots,x_n \}$
\ElsIf{$l>k$}
\State Increase the position of individuals $x_{k+1},\dots,x_{l-1}$ by one;
\State Assign the possible positions $k,k+1$ to individual $x_k$ and its copy with probability $1/2$; 
\State $P'\leftarrow \{x_1,\dots,x_k\textnormal{v}x_{k}',x_k\textnormal{v}x_{k}',x_{k+1},\dots,x_{l-1},x_{l+1},\dots,x_n \}$
\Else
\State Replace $x_k$ by $x_k'$;
\State $P'\leftarrow \{x_1,\dots,x_{k-1},x_{k}',x_{k+1},\dots,x_n \}$
\EndIf
\Ensure New population $P'$ of size $N$

\end{algorithmic}

\end{algbox}
}

\begin{figure}
\begin{center}
\includegraphics[scale=.325]{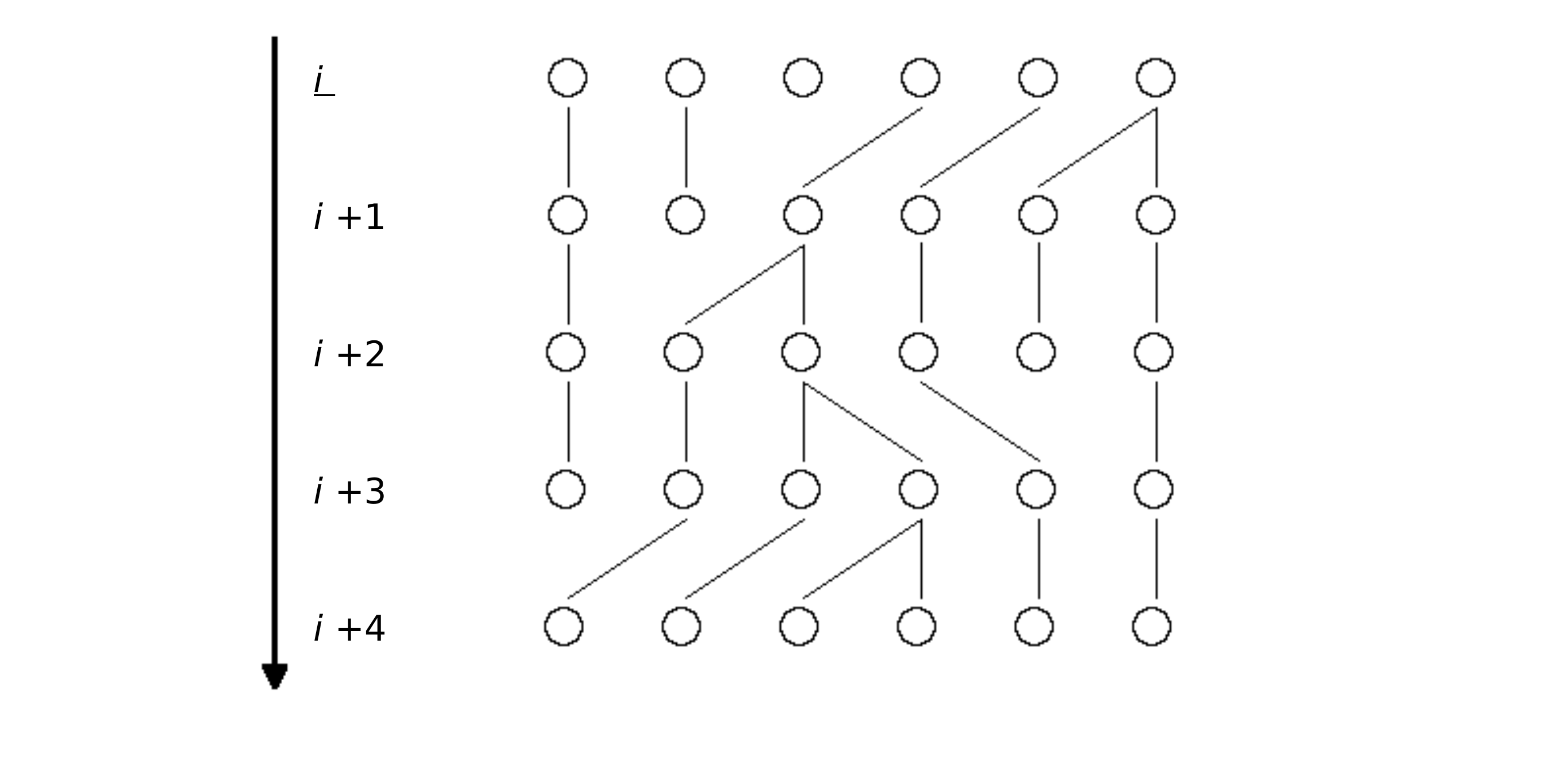}
\end{center}
\caption{Iterations of a Moran process, where the order is maintained according to~\ref{proc:pmoran}.}
\label{fig:plane}
\end{figure}

Since the time $i^*$ is almost surely finite, we assume in what follows that at $i=0$ there already exists a genealogy $T_0$. We  disregard the branch lengths of $T_0$ and instead assume that internal nodes of $T_0$ are labeled by integers $1,\dots,n-1$ respecting the order of the past split events which they represent, and obtain an $n$-sized Yule tree, where the leaves $\iota_1,\dots \iota_n$ represent the individuals $x_1,\dots,x_n$ of the population. By \cite{aldous:probability}, $T_0$ can be treated as the result of a Yule process of $n-1$ iterations. In the following iterations $i>0$, $M$ can be emulated by observing the genealogy $T_i,i\geq 0$ directly, where $T_i$ is modified according to Procedure~\ref{proc:emg}.

\vbox{
\vspace{5pt}
\begin{algbox}{Evolving Moran Genealogy given $T_i$}{proc:emg}
\begin{algorithmic}[1]
\Require Current tree $T_i\in\mathcal{T}_n$
\State Choose $k,l\in \{1,\dots,n\}$ uniformly and independently
\If{$k=l$}
\State $T_{i+1}\leftarrow T_i$
\Else
\State Remove the leaf representing the killed individual $x_l$ alongside the branch connecting it to the remainder of $T_i$, and the internal node $\nu_j$ at the position in $T_i$ that branch is attached at;
\State Merge the two branch segments $b,b'$ connected to $\nu_j$ into one;
\State Merge all pairs of branch segments $b,b'$ in layers $j,j+1$ belonging to the same branch into single branch segments;\\ \Comment{4-6 "remove" layer $j+1$}
\State Decrease the labels of $\nu_{j'},j'>j$ by one in $T_i$;
\State Turn the leaf representing the duplicated individual $x_k$ in $T_i$ into an internal node with two new leaves appended to it and label it by $n-1$; $T_{i+1}\rightarrow T_i$;
\EndIf
\Ensure New tree $T_{i+1}\in\mathcal{T}_n$
\end{algorithmic}

\end{algbox}
}

Let $\Phi_{k,l}(T_i)$ denote the output of Procedure~\ref{proc:emg} given $k,l$. Then, we define:\\
\vbox{
\begin{definition}
\vskip 0pt ~
\begin{itemize}
\item The process $(\mathcal{T}_i)_{i\in\mathbb{N}}$
with $T_{i+1}=\Phi_{k,l}(T_i)$ for uniform $k,l$ and uniform $T_0$ is called \textit{Evolving Moran Genealogy}, for short \textit{EMG}.
\item We will identify a leaf $\iota$ of any $T_i,i\geq 0$ with the individual $x$ representing it and write $x\in T_i$ if an individual $x$ is part of $T_i$.
\item For  $T,T'\in\mathcal{T}_n$, we define the notation
$$T \rightarrow T' \Leftrightarrow \exists k,l\in\{1,\dots,n\}: \Phi_{k,l}(T)=T'$$
\end{itemize}
\end{definition}
}
\begin{figure}
\includegraphics[scale=.25]{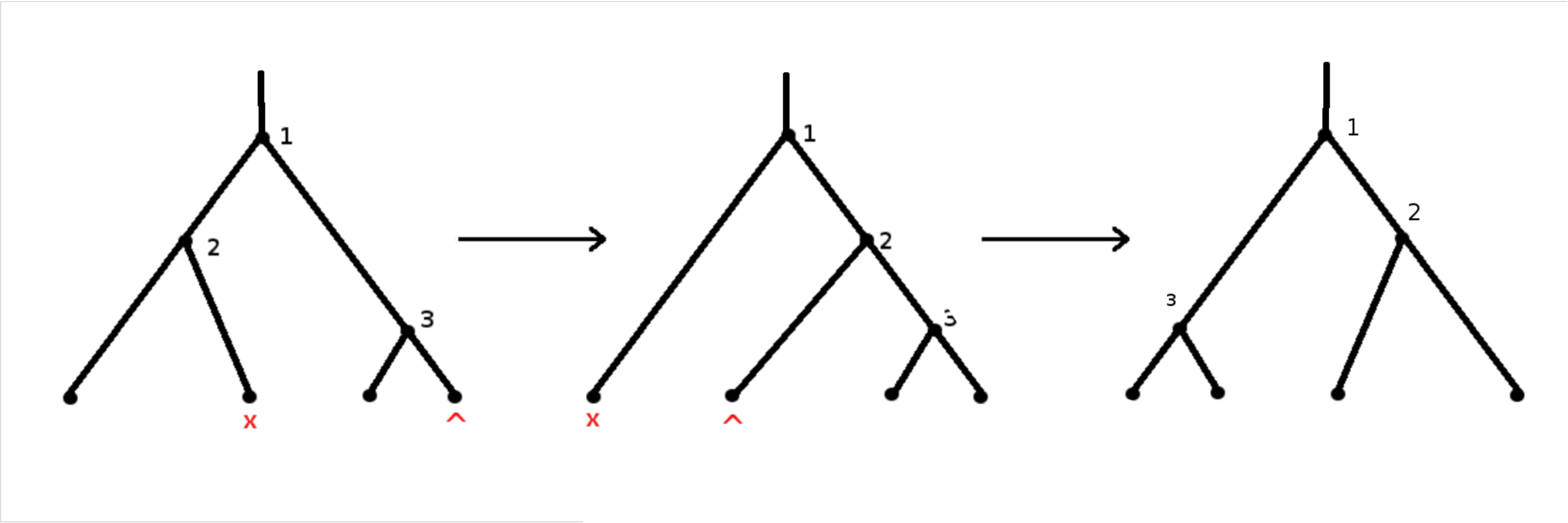}
\caption{Two steps in an \textit{EMG} of size $4$. Each step, one individual is killed ("\textcolor{red}{X}") and one duplicated ("\textcolor{red}{$\Lambda$}").}
\label{fig:emg2step}
\end{figure}
The \textit{EMG} (see Figure~\ref{fig:emg2step}) is a Markov chain on the set $\mathcal{T}_n$. Its transition matrix $E$ has nonzero diagonal entries, since $l=k$ entails $T_{i+1}=T_i$. Secondly, at most $n-1$ transition steps of the above form are needed to transform $T_i$ into some arbitrary $T'\in\mathcal{T}_n$, because all duplications may be applied to one single individual $x$ of the population $\mathcal{P}_i$ at time $i$ and its descendants, and all removals to the remaining individuals of $\mathcal{P}_i$, corresponding to a Yule process of $n-1$ iterations. As a result, the \textit{EMG} is a recurrent Markov chain. Aperiodicity of the \textit{EMG} also follows, as $T_{i+1}=T_i$ can always happen with positive probability.\\
As a consequence, there exists a stationary distribution $P^*$ of the \textit{EMG} on $\mathcal{T}_n$. Since we may interpret the genealogy $T_{i}$ as a result of a Yule process at each stage $i>0$, and since all $T\in\mathcal{T}_n$ are equally likely under the Yule process, it follows that $P^*$ is the uniform distribution, i.e. $P^*(T)=\frac{1}{(n-1)!}$ for all $T\in\mathcal{T}_n$.\\ 
The relation $T\rightarrow T$, indicating that $T$ can be transformed into $T'$ by some duplication/remove combination in $M$, can be used to describe the entries of $E$. Importantly,
\begin{equation}\notag
T=T'\Rightarrow T\rightarrow T';\ T\rightarrow T'\not\Rightarrow T' \rightarrow T
\end{equation}
Then the $T,T'$-entry of $E$ can be denoted in the following way:
\begin{equation}\notag
\Pr(T_{i+1}=T'|T_i=T)=\begin{cases}
0 & T\not\rightarrow T' \\
\frac{|\{(k,l)\in\{1,\dots,n\}^2:\Phi_{k,l}(T)=T'\}|}{n^2} & \textnormal{otherwise} \\
\end{cases}
\end{equation}
In particular, the diagonal entries of $E$ are nonzero and depend on $T$. Suppose $n=2^k$ for some $k\geq 2$, and consider the \textit{caterpillar} $C\in\mathcal{T}_n$ obtained under the Yule process by always choosing the leftmost individual to split, and any \textit{complete binary search tree} $B\in\mathcal{T}_n$, i.e. a tree characterized by the fact that there is an equal number of leaves on both subtrees below each internal node. Then $\Pr(T_{i+1}=C|T_i=C)=\frac{2n}{n^2}$, whereas $\Pr(T_{i+1}=B|T_i=B)=\frac{n+2}{n^2}$.
\section{Results}
\subsection{The Process of Tree Balance}
Since a tree $T\in\mathcal{T}_n$ obtained under the \textit{EMG} is plane and individuals ordered from left to right, we may consider the left and right subtrees $T^l,T^r$ below the root node $\nu_1$. Essentially, $T^l$ can be thought of as the induced subtree of all leaves on the left side below $\nu_1$ (the same holds for $T^r$). Suppose we are interested in the dynamics of the number of leafs on the left, i.e. $|T^l|$.\\
\vbox{
\begin{definition}
The process $\textit{TB}:=(|T_i^l|)_{i\in\mathbb{N}}$ is called Tree Balance Process of the Evolving Moran genealogy.
\end{definition}
}
The choice between observing left and right subtree size is arbitrary, since always $|T^r|=n-|T^l|$. Closely related to the process $\textit{TB}$ is the $\Omega_1$-statistic \cite{aldous:balance,feretti:srecevents} observed over time, where $\Omega_1(T_i):=\min(|T_i^l|,|T_i^r|)$, and one observes $(\Omega_1)_{i\in\mathbb{N}}$. There is little difference between $\textit{TB}$ and the process of $\Omega_1$, as paths of $\textit{TB}$ are essentially mirrored at $\frac{n}{2}$ when considering $\Omega_1$. Determining the dynamics of $\textit{TB}$ thus suffices to also obtain those of $\Omega_1$.\\
\begin{prop}
The transition probabilities of $\textit{TB}$ are as follows:\\
\hfill\\
If $2\leq |T_i^l|\leq n-2$,
\begin{equation}\notag
\Pr\left(|T_{i+1}^l|=\omega ~\Bigm| ~|T_i^l|\right)=
\begin{cases}
\frac{|T_i^l|(n-|T_i^l|)}{n^2} & \omega = |T_i^l|+1 \\
\frac{|T_i^l|^2+(n-|T_i^l|)^2}{n^2} & \omega = |T_i^l|\\
\frac{|T_i^l|(n-|T_i^l|)}{n^2} & \omega = |T_i^l|-1 \\

\end{cases}
\end{equation}
If $|T_i^l|=1$,
\begin{equation}\notag
\Pr\left(|T_{i+1}^l|=\omega\right)=
\begin{cases}
\frac{1}{n} & \omega = 2 \\
\frac{(n-1)^2+2}{n^2} & \omega = 1\\
\frac{1}{n^2} & \textnormal{otherwise} \\
\end{cases}
\end{equation} 
And if $|T_i^l|=n-1$,
\begin{equation}\notag
\Pr\left(|T_{i+1}^l|=\omega\right)=
\begin{cases}
\frac{1}{n} & \omega = n-2 \\
\frac{(n-1)^2+2}{n^2} & \omega = n-1\\
\frac{1}{n^2} & \textnormal{otherwise} \\
\end{cases}
\end{equation} 
\end{prop}
\begin{proof}
Suppose $2\leq |T_i^l|\leq n-2$. $|T_{i+1}^l|=|T_i^l|-1$ is the case if one individual on the left side is removed and one on the right is duplicated. This happens with probability $\frac{|T_i^l|(n-|T_i^l|)}{n^2}$. We obtain the same probability for the case $|T_{i+1}^l|=|T_i^l|+1$.\\ 
Finally, we have $|T_{i+1}^l|=|T_i^l|$ if removal and duplication take place on the same side. The probability of this is $\frac{|T_i^l|^2+(n-|T_i^l|)^2}{n^2}$.\\ 
The only difference in the cases $|T_i^l|=1,n-1$ is that one has to include the possibility of a complete removal of $T_i^l$ in the first and $T_i^r$ in the latter case. If this happens, $|T_{i+1}^l|$ and $|T_{i+1}^r|$ are independent of $|T_i^l|$ and $|T_i^r|$. In fact, $|T_{i+1}^l|$ then assumes any value $1,\dots,n-1$ with uniform probability \cite{aldous:balance}.\\
Therefore, considering $|T_i^l|=1$, the total probability $\Pr(|T_{i+1}^l|=2)$ is the sum of the probability that the individual on the left is duplicated and one on the right is removed, which amounts to $\frac{(n-1)}{n^2}$, and the probability that it is removed and $|T_{i+1}^l|=2$ by chance, which happens with probability $\frac{n-1}{n^2}\frac{1}{n-1}=\frac{1}{n^2}$. This sum equals $\frac{1}{n}$.\\ 
For $\Pr(|T_{i+1}^l|=1)$ and $\Pr(|T_{i+1}^l|>2)$ the calculation is similar, and of course the case $|T_i^l|=n-1$ can be treated analogously.
\end{proof}
We notice that the transition probabilities within an episode are identical to those in a two-allele Moran Model. In the large-population limit, tree balance has been identified before as a Wright-Fisher Diffusion \cite{pfaffelhuber:mrca,delmas:families}, to which $(|T_i^l|/n)_{i\in\mathbb{N}}$ (the values of \textit{TB} divided by population size) indeed converges if time is rescaled by $2/n^2$. The transition probabilities of \textit{TB} in the cases $|T_i^l|=1$ and $|T_i^l|=n-1$ correspond to the behaviour of the tree balance process in the limit, which, upon hitting a "boundary" ($0$ or $1$), resets at any value in $[0,1]$ with uniform probability.\\ 
Also, if $n$ is large and $\frac{|T_i^l|}{n}$ is either close to $0$ or $1$ (the genealogy is "unbalanced"), the strength of diffusion is weakest. Consequently, if the Evolving Moran Genealogy enters an unbalanced state, genealogies in the following generations are expected to be unbalanced as well.It might be worth to investigate whether and to what extent this can have diluting effects on statistics like Tajima's $D$ \cite{tajima:d}, or $T_3$ \cite{wieheli:t3}, which crucially depend on tree shape; particularly because unbalanced tree shapes can introduce a bias in the frequency spectrum of polymorphisms that may be confused with the effect of of natural selection.\\
We may refer to the phases between complete removals of $T_i^l$ or $T_i^r$ as \textit{episodes} of the process \textit{TB}. It is worth mentioning that the complete removals of $T_i^l$ or $T_i^r$, i.e. starting and ending times of episodes in $\textit{TB}$, are precisely the times of \textit{MRCA} jumps in the \textit{EMG}, to be discussed in section~\ref{sec:mrca}.
\subsection{Time Reversal of the \textit{EMG}}
We consider a second process on the set $\mathcal{T}_n$. Let $T\in\mathcal{T}_n$ and consider the \textit{Merge-Regraft} operation:\\
\vbox{
\vspace{5pt}
\begin{algbox}{Merge-Regraft on given $T$}{proc:mrg}
\begin{algorithmic}[1]
\Require $T\in\mathcal{T}_n$
\State Choose one branch segment $b$ of $T$ from the set $\{b_1,\dots,b_{\frac{n(n+1)}{2}}\}$ with probability
\begin{equation}\notag
\Pr(b=b_k)=\begin{cases}\frac{1}{n^2} & b_k\textnormal{ ends in a leaf}\\
\frac{2}{n^2} & \textnormal{ otherwise}
\end{cases}
\end{equation} 

and $\chi$ from $\{\textit{left,right}\}$ with equal probability;
\If{$b$ ends in a leaf}
\State $T'\leftarrow T$
\Else
\State Remove the $n$-th layer of $T$; remove $\nu_{n-1}$; place leafs at the tips of the branch segments that extend across layer $n-1$; \\ \Comment{the position of $\nu_{n-1}$ is then occupied by some leaf}
\State Regraft a new leaf at branch $b$ with orientation $\chi$ in $T$ according to Procedure~\ref{proc:rg} (skipping step 1);
\State $T'\leftarrow T$
\EndIf
\Ensure $T'\in\mathcal{T}_n$
\end{algorithmic}

\end{algbox}
}
Graphically, this operation adds a new leaf by performing a regrafting operation at segment $b$, and merges the two leaves belonging to the lowermost split into one. Note that if $b$ ends in a leaf ($l(b)=n$), regrafting at this branch establishes the lowermost split, in which case the tree remains the same. To clarify the purpose of step $5$, one may also imagine that all leafs are moved up by one layer, such that two of them must "merge". An example of such an operation is shown in Figure~\ref{fig:mr}.\\
\begin{figure}
\begin{center}
\includegraphics[scale=.375]{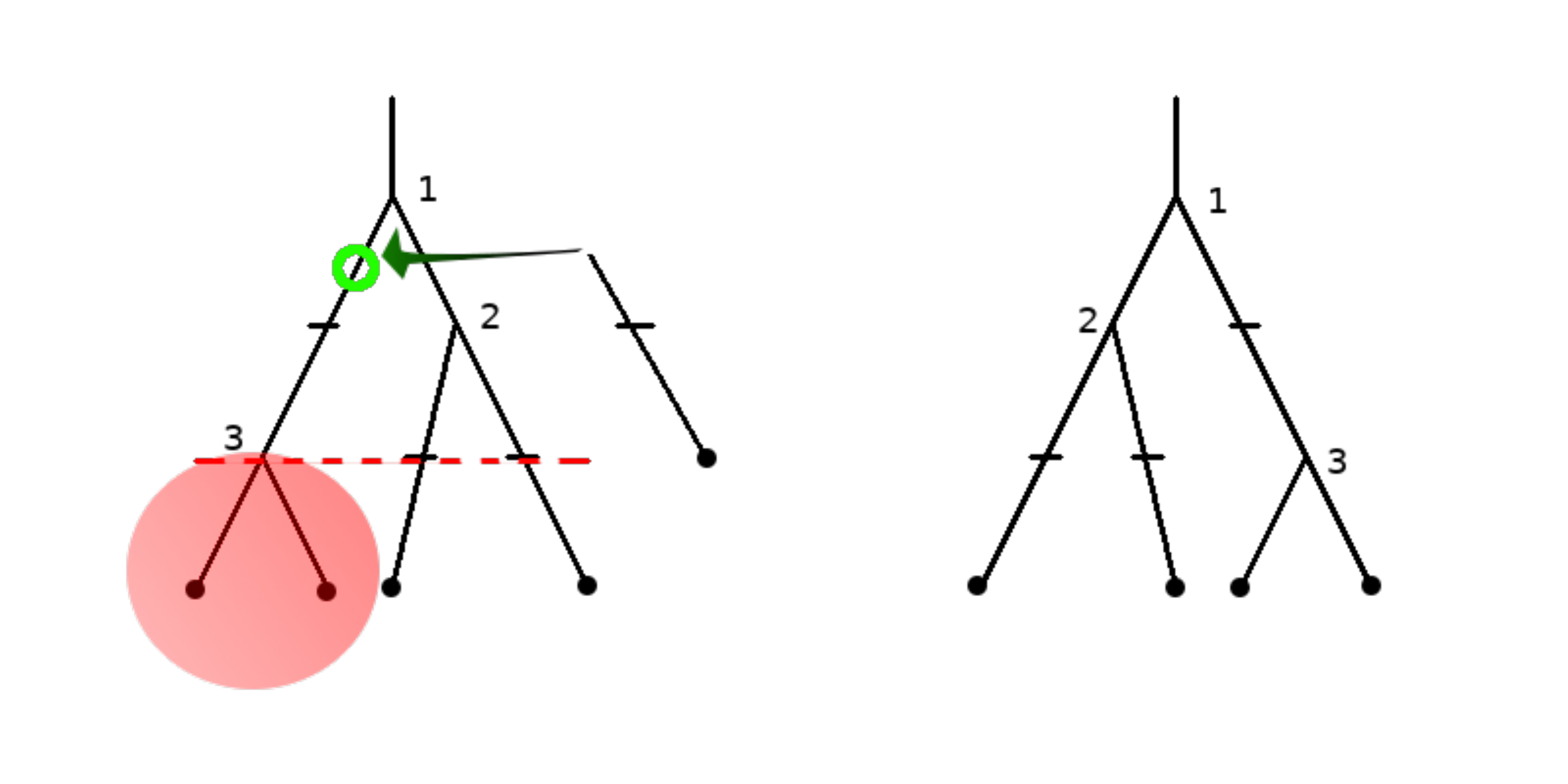}
\end{center}
\caption{Visualization of a possible Merge-Regraft operation on the tree on the left. The pair of leaves belonging to the lowermost split merge (red area), the lowermost layer is removed, and a branch ending in a single "revived" leaf is inserted at the branch segment carrying the "\textcolor{green}{o}" mark. The resulting tree is shown on the right; note that both trees belong to $\mathcal{T}_4$.}
\label{fig:mr}
\end{figure} 
Let the result of this operation be denoted by $\Phi'_{b,\chi}(T)$. $\Phi'_{b,\chi}(T)$ is itself an object of the set $\mathcal{T}_n$. The function $\Phi'$ can be thought of as a combinatorial inversion of $\Phi$: The split event facilitated by $\Phi$ is revoked by $\Phi'$, and the leaf that is removed under $\Phi$ can be recovered ("revived") by $\Phi'$ by regrafting; in fact, we have $T\rightarrow T'\Leftrightarrow \exists (b,\chi): \Phi'_{b,\chi}(T')=T$.\\
We consider the process 
$$R:=(\tilde{T}_{i})_{i\in\mathbb{N}},$$ 
where $\tilde{T}_0$ is uniformly chosen and, given $\tilde{T}_i$, $\tilde{T}_{i+1}$ is generated by the mechanism described above, i.e. $\tilde{T}_{i+1}=\Phi'_{b,\chi}(\tilde{T}_i)$ for some random choice of $b$ and $\chi$ (See Figure~\ref{fig:emgb}). In what follows, we will show that this process represents a \textit{time-reversal} of the $EMG$ (where the term time-reversal is used in th	e sense of e.g. \cite{lovasz:reversal}).\\
\begin{lemma}
For all $T,T'\in\mathcal{T}_n$:
\begin{equation}
\textnormal{Pr}_{\textit{EMG}}(T_{j+1}=T'|T_{j}=T)=\textnormal{Pr}_{R}(\tilde{T}_{i+1}=T|\tilde{T}_{i}=T')
\label{eq:transprobabilities}
\end{equation}
with $\Pr_{\textit{EMG}}$ denoting the transition probability of the \textit{EMG}-process, and $\Pr_{R}$ that of the process $R$.
\end{lemma}
\begin{proof}
Recall that $\Phi$ was dependent on the choice of $k,l$, which were both chosen uniformly. The probability of $k=l$ in one step of the \textit{EMG}, which for all $T\in\mathbb{T}_n$ entails $\Phi_{k,k}(T)=T$, is $\frac{1}{n}$. The probability that the branch segment $b$ chosen in a transition of the process $R$ is inside layer $n$, which always leads to $\Phi'_{b,\chi}(\tilde{T})=\tilde{T}$, is also $\frac{1}{n}$ in total for any $T$.\\
We define for arbitrary $T,T'\in\mathcal{T}_n,T\rightarrow T'$:
\begin{align*}
S_1&:=\{(k,l)\in\{1,\dots,n\}^2:\Phi_{k,l}(T)=T',k\neq l\}\\
S_2&:=\{(b,\chi)\in\{b_1,\dots,b_{\frac{n(n-1)}{2}}\}\times\{\textit{left,right}\}:\Phi'_{b,\chi}(T')=T\}
\end{align*}
If we can show that $|S_1|=|S_2|$, we are done. Let $(k,l)\in S_1$. Let $\nu$ denote the internal node deleted by $\Phi_{k,l}(T)$. Choosing the regrafting site $b$ as the branch segment generated by merging the two segments connected to $\nu$ (compare step 5 in Procedure~\ref{proc:emg}), and $\chi$ according to whether the branch of $x_l$ extends to the left or right in $T$, we obtain a unique $(b,\chi)\in S_2$, which yields a mapping $\mu:S_1\rightarrow S_2$. Since by definition of the Yule process there cannot be two or more tuples $k,l$ and $k',l$ with $k\neq k'$ such that $\Phi_{k,l}(T)=\Phi_{k',l}(T)$, $\mu$ is injective.\\ 
On the other hand, for any $(b,\chi)\in S_2$ such that $\Phi'_{b,\chi}(T')=T$, choosing $l$ such that $x_l$ is the leaf regrafted in $T'$ by $\Phi'_{b,\chi}(T')$ and $k$ such that $x_k$ is the leaf replacing the highest-labeled internal node in $T'$ by $\Phi'_{b,\chi}(T')$ (see step 5 of Procedure~\ref{proc:mrg}), we obtain $(k,l)\in S_1$ such that $\mu((k,l))=(b,\chi)$. Therefore, $\mu$ is a bijection and both sets are equally large.
\end{proof}
\begin{corollary}
The process $R$ represents the time-reversed process of the \textit{EMG}. 
\end{corollary}
\begin{proof}
The existence of a time-reversed process $R(\textit{EMG})$ on $\mathcal{T}_n$ is provided by the fact that it is a recurrent Markov chain with nonzero stationary distribution. The transition probabilities of this process are
\begin{align*}
\textnormal{Pr}_{R(\textit{EMG})}(T_{j+1}=T'| T_{j}=T)&=\textnormal{Pr}_{\textit{EMG}}(T_{i}=T'| T_{i+1}=T)\\
&=\textnormal{Pr}_{\textit{EMG}}(T_{i+1}=T| T_{i}=T')\frac{\textnormal{Pr}_{\textit{EMG}}(T_i=T')}{\textnormal{Pr}_{\textit{EMG}}(T_{i+1}=T)}\\
\end{align*}
Since the stationary distribution of the \textit{EMG} is the uniform distribution, we have $\textnormal{Pr}_{\textit{EMG}}(T_i=T')=\textnormal{Pr}_{\textit{EMG}}(T_{i+1}=T)$. Therefore, 
\begin{equation}
\textnormal{Pr}_{R(\textit{EMG})}(T_{j+1}=T'| T_{j}=T)=\textnormal{Pr}_{\textit{EMG}}(T_{i+1}=T| T_{i}=T')
\end{equation}
and these transition probabilities are exactly the ones provided by the process $R$ (compare equation~\eqref{eq:transprobabilities}).
\end{proof}
\vbox{
\begin{definition}
We call the process $R$ the \textit{Evolving Moran Genealogy backwards in time}, for short $\textit{EMG}^\flat$.
\end{definition}
}
\begin{figure}
\includegraphics[scale=.25]{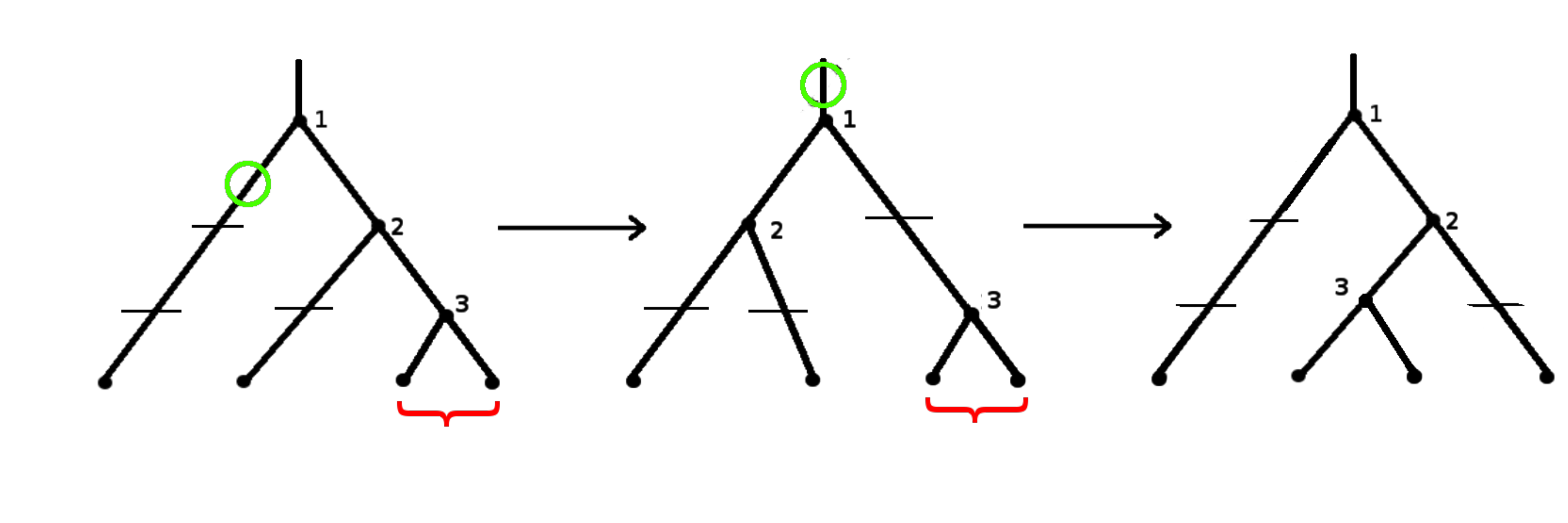}
\caption{Two possible transitions of the process $R$, i.e. the $\textit{EMG}^\flat$ of size $4$. Branch segment chosen for regrafting marked by "\textcolor{green}{o}".}
\label{fig:emgb}
\end{figure}
We end this section with the remark that the \textit{EMG} and $\textit{EMG}^\flat$ bear a certain resemblance to the \textit{Aldous chain} on cladograms \cite{aldous:cladograms}, of which also an infinite-size limiting process has been described \cite{lohr:aldouschain}. The state space of the Aldous chain consists of cladograms, which are constructed slightly differently than $\mathcal{T}_n$. In particular, they are not internally labelled and their branches are not subdivided into segments. Still, both \textit{EMG} and $\textit{EMG}^\flat$ feature one of the two components of the \textit{Aldous move}, i.e. removal (\textit{EMG}) and reattachment ($\textit{EMG}^\flat$) of a single leaf and its associated branch. In combining those two mechanisms, one can define a third Markov chain, which would represent a version of the Aldous chain on $\mathcal{T}_n$. A possibility of further research could be to compare the mixing times of the respective chains.  
\subsection{\textit{MRCA} and Age Structure in the $\textit{EMG}^\flat$}\label{sec:mrca}
Besides the technical aspects, there are some reasons why the $\textit{EMG}^\flat$ as a stochastic process can prove useful in theoretical and practical regard. While the transitions in the $\textit{EMG}$ rely on two random mechanics (duplication and removal), in the $\textit{EMG}^\flat$ they are unified within the regrafting operation. Because of that, aspects about the genealogy itself may become more tractable to analytic investigation. One good example for this is the \textit{MRCA}-process.\\
Let $x^*$ denote the \textit{MRCA} of a genealogy generated by a neutral Moran process. With probability $1$, after some finite time a descendant of $x^*$ will become ancestral to the entire population, establishing a new \textit{MRCA}. Therefore, defining $\chi_{\textit{MRCA}}(i)=1$ if at time $i$ a new \textit{MRCA} of the population is established (i.e., the \textit{MRCA} "jumps", which in the \textit{EMG} is represented by the eventual obliteration and repositioning of the root node of $T_i$), and $\chi_{\textit{MRCA}}(i)=0$ otherwise, we call $(\chi_{\textit{MRCA}}(i))_{i\geq 0}$ the \textit{MRCA}-process.\\
\begin{lemma}\label{lemma:mrca1}
$(\chi_{MRCA}(i))_{i\geq 0}$ is a geometric jump process of intensity $\frac{2}{n^2}$.
\end{lemma}
\begin{proof}
In the $\textit{EMG}^\flat$, the root of the genealogy $T_i$ changes if and only if the imaginary branch $b_1$ is chosen for regrafting. This happens with probability $\frac{2}{n^2}$ in each step (see also Figure \ref{fig:jumps}).
\end{proof}
This agrees with the results in \cite{pfaffelhuber:mrca}, where the \textit{MRCA}-process in the infinite-population limit is identified as a Poisson-process of intensity $1$, which is the limit of the geometric jump process as $n\rightarrow \infty$ with time sped up by $\frac{n^2}{2}$. Also by Lemma \ref{lemma:mrca1}, the number of steps needed to observe any number $r\in\mathbb{N}$ of root jumps follows a negative binomial distribution $\textit{NB}(r,\frac{2}{n^2})$.\\
The discrete structure of the $\textit{EMG}^\flat$ also allows us to derive properties of the \textit{MRCA}-process during ongoing  fixations in the underlying Moran process.\\
\vbox{
\begin{definition}
Suppose a member $\tilde{x}$ in a neutral Moran process was generated as the result of some duplication at time $i^*\gg 0$, and we observe, by chance, the fixation of descendants of $\tilde{x}$ in the population at time $i^{\textit{fix}}\gg 0$.\\
\begin{itemize}
\item
$i^{\textit{fix}}$ is called \textit{fixation time} of the individual $\tilde{x}$
\item $i^*$ is called \textit{birth time} of $\tilde{x}$
\end{itemize}
\end{definition}
}
\begin{lemma}\label{lemma:fixjumps}
In a Moran Population of size $n\geq 2$, we expect to observe $2-\frac{2}{n}$ \textit{MRCA}-jumps between (and including) $i^*$ and $i^{\textit{fix}}$.
\end{lemma}
\begin{proof}
By our assumptions, we know that one \textit{MRCA}-jump necessarily happens at the transition of $T_{i^{\textit{fix}}-1}$ to $T_{i^{\textit{fix}}}$. We claim that we expect another one during the remainder of the fixation time.\\
Let $l=i^{\textit{fix}}-1-i^*$. We know $l\geq n-2$, since the minimal number of steps necessary to fix the descendants of $\tilde{x}$ is $n-1$. The sequence of genealogies $(T_{i^*},\dots,T_{i^{\textit{fix}}-1})$ in reverse order is a path $y=(T_0',\dots,T_l')$ of the $\textit{EMG}^\flat$, where $T_{0}'=T_{i^{\textit{fix}}-1},\dots,T_l'=T_{i^*}$.\\ 
The set of $\textit{EMG}^\flat$-time steps $\{1,\dots,l\}$ contains a subset $I=\{i_1,\dots,i_{n-1}\}$, $i_l\leq i_{l+1}$, where $i\in I$ if and only if $x\in T_l'$ holds for the individual $x$ regrafted at time $i$; i.e., $x$ is also present in the population at the time $l$, which represents the birth time of $\tilde{x}$ in the $\textit{EMG}^\flat$. $I$ thus consists of exactly the times where individuals of the population at the birth time of $\tilde{x}$ are revived. In particular, $i_1=0$ and $i_{n-1}=l$. For $i\in I$, let $S_i:=\{x \in T'_i:x\in T'_l\}$ denote the set of individuals that will be members of the population at the birth time of $\tilde{x}$. \\
Starting from $T_{1}'$, a root jump can only occur in some step $i$ if $i\in I$. For any $i_j\in I$, we know that regrafting must take place in some layer $k\leq j+1$. We therefore consider the sequence 
$$\hat{T}^{(1)} = (T_l')_{S_{i_1}},\dots,\hat{T}^{(n-2)}=(T_l')_{S_{i_{n-2}}},\hat{T}^{(n-1)}=T_l'$$
of $S_i$-induced subtrees of $T_l'$ for $i\in I$. Since $T_{l}'$ is a Yule tree, by Lemma~\ref{lemma:pwr} we may assume that each $\hat{T}^{(j)}$, $j=2,\dots,n-1$ is obtained from a random grafting operation~\ref{proc:rg} performed on $\hat{T}^{(j-1)}$. The probability of a root jump in step $i_j$ is therefore the probability of regrafting at the imaginary branch of $\hat{T}^{(j)}$, which equals $\frac{2}{j(j+1)}$.\\
The total expected number of root jumps along the $\textit{EMG}^\flat$-path $y$ is then
$$\sum_{k=2}^{n-1}\frac{2}{k(k+1)}=\frac{n-2}{n}.$$
This expression equals $1-\frac{2}{n}$. Adding the root jump that necessarily occurs in step $1$, we end up with an expectation of $2-\frac{2}{n}$.

\end{proof}
\begin{figure}

\includegraphics[scale=.25]{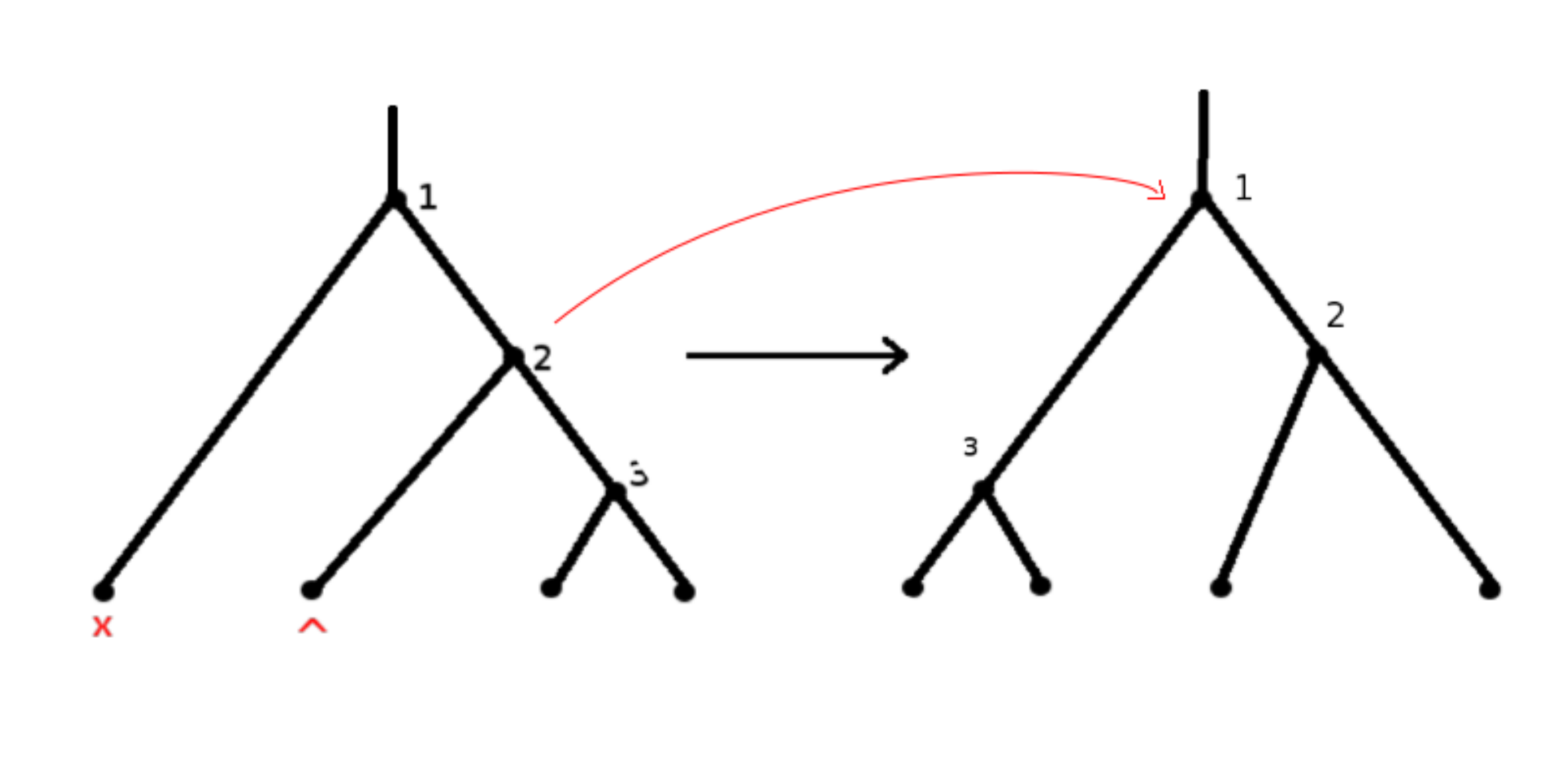}
\includegraphics[scale=.25]{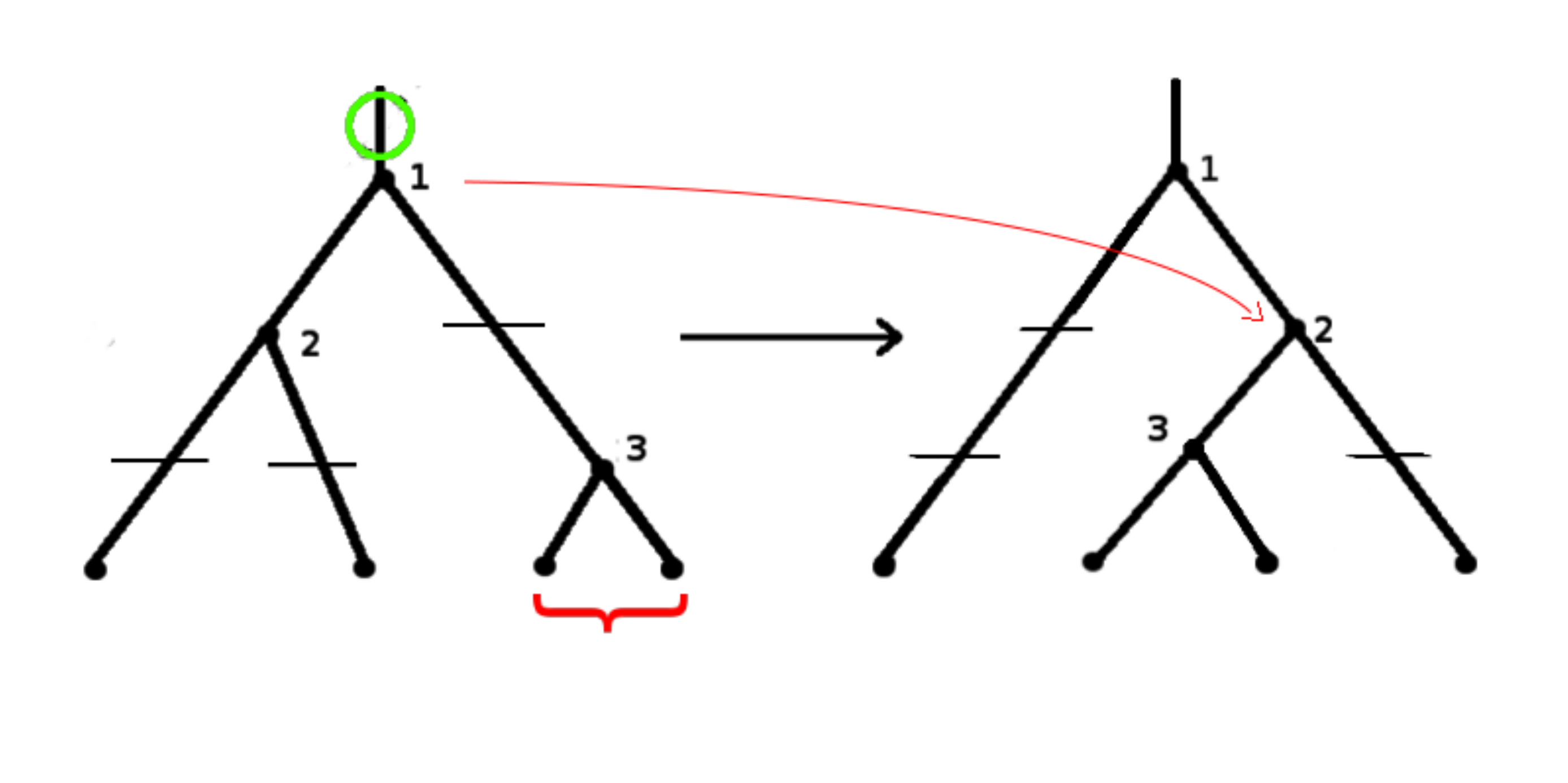}
\caption{\textit{MRCA} jumps in the \textit{EMG} and $\textit{EMG}^\flat$ (See also Figures~\ref{fig:emg2step}, \ref{fig:emgb}).}
\label{fig:jumps}
\end{figure}

Considering the infinite-population limit, we conclude that between birth and fixation time of an individual, there are $2$ expected \textit{MRCA} jumps in total. It is noteworthy that in the backward view, we do not need to condition the process on fixation of the allele $\tilde{x}$, since every root jump necessarily coincides with a fixation. In the limiting process, one possible solution to this problem is to consider a Wright-Fisher diffusion conditioned on not hitting $0$, which is achieved by introducing an artificial drift term (see \cite{delmas:families}).\\
By similar arguments, we may calculate the exact distribution of root jumps during a neutral fixation for any $n$, and show that these distributions converge as $n\rightarrow \infty$. For $n\geq 2$, let $\Pr_n(k)$ denote the probability of observing $k$ root jumps during a neutral fixation in an $\textit{EMG}$ of size $n$, and $\Pr_\infty(k)$ the same probability in the infinite-population limit. $\Pr_n(k)$ can be written as follows:
$${\Pr}_n(k):=\sum_{2\leq i_1,\dots,i_{k-1}\leq n-1}\Pi_1^k\frac{2}{i_k(i_k+1)}\Pi_{j\neq i_1,\dots,i_{k-1}}\left(1-\frac{2}{j(j+1)}\right)$$
This is obtained by multiplying the probabilities of regrafting at the imaginary branches of $\hat{T}^{(i_1)},\dots,\hat{T}^{(i_k-1)}$ (in the sense of the notation used in Lemma~\ref{lemma:fixjumps}) and not regrafting at the imaginary branches of all other $\hat{T}^{(j)}$, summed up over all possible choices of $i_1,\dots,i_{k-1}$. For $k=1$, in which case the imaginary branch is never chosen for regrafting, we have simply:
\begin{equation}\label{eq:p1}{\Pr}_n(1)=\Pi_{j=2}^{n-1}\left(1-\frac{2}{j(j+1)}\right)=\Pi_{j=2}^{n-1}\frac{(j+2)(j-1)}{j(j+1)}=\frac{n+1}{3(n-1)}\end{equation}
By reordering of the factors, we obtain the following expressions for $k=2,3,\dots$:
\begin{align}\label{eq:jumpdist}
{\Pr}_n(2)&=\Pi_{j=2}^{n-1}\left(1-\frac{2}{j(j+1)}\right)\left[\sum_{k=2}^{n-1}\frac{2}{k(k+1)-2}\right]\notag\\
{\Pr}_n(3)&=\Pi_{j=2}^{n-1}\left(1-\frac{2}{j(j+1)}\right)\left[\sum_{k=2}^{n-2}\frac{2}{k(k+1)-2}\cdot\left(\sum_{l=k+1}^{n-1}\frac{2}{l(l+1)-2}\right)\right]\notag\\
\dots & \notag\\
\end{align}
For small $k$, it is possible to also obtain closed-form expressions for $\Pr_n(k)$ similar to equation~(\ref{eq:p1}) using computational algebra. For $n=2$, $\Pr_n(1)=1$, and as $n\rightarrow \infty$, $\Pr_n(1)$ converges to $\frac{1}{3}=:\Pr_\infty(1)$, decreasing monotonously. Note that this can be interpreted as an analogon to a result in \cite{pfaffelhuber:mrca} about the infinite-population limit. In the terminology of this work, the value $\frac{1}{3}$ corresponds to the probability that the "next fixation curve has not yet started" at the time $i^*$.\\
The other probabilities in the infinite-population limit can be calculated numerically by evaluating the infinite-sum expressions on the right-hand sides of~\eqref{eq:jumpdist}. By continuity, the probabilities $\sum_{k=1}^{\infty}\Pr_\infty(k)$ sum up to $1$. The largest contribution comes from $\Pr_\infty(2)=\frac{11}{27}=\frac{11}{9}\cdot\frac{1}{3}$. As a side note, since $\Pr_n(2)$ increases monotonously with $n$, we can calculate that for $n\leq 9$, the distribution is dominated by $\Pr_n(1)$, whereas for $n\geq 10$, the probability $\Pr_n(2)$ provides the largest value. Figure~\ref{fig:jumpdist} outlines some of the distributions for different population sizes.\\
\begin{figure}
\includegraphics[scale=.625]{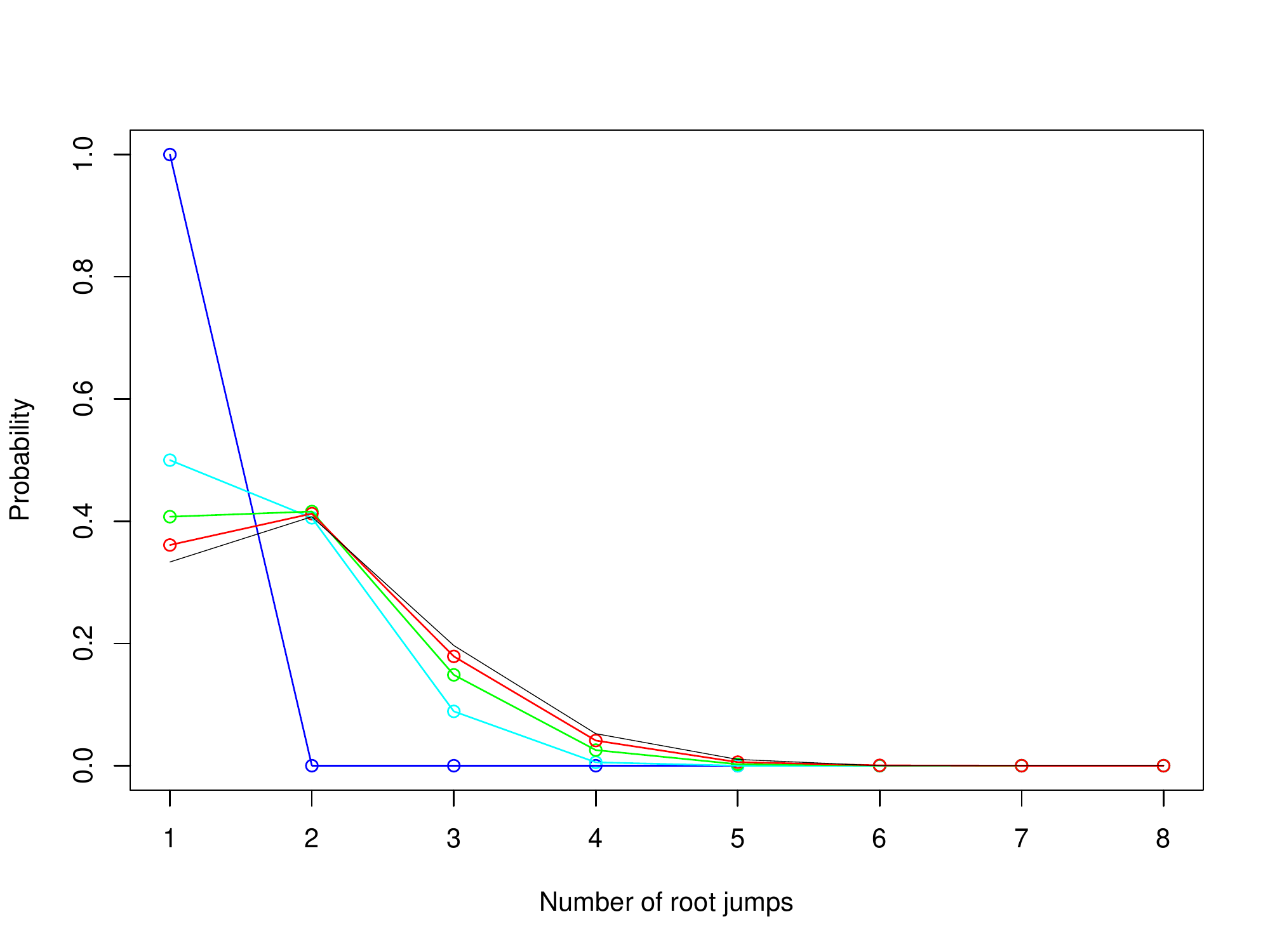}
\caption{The distributions of $P_n(k)$, $k=1,\dots,8$, $n=2$ (blue), $5$ (turquoise), $10$ (green), $25$ (red) and $\infty$ (black).}
\label{fig:jumpdist}
\end{figure}
Another implication of the $\textit{EMG}^\flat$ is that coalescent events are "visible" in the genealogy for a certain average number of steps, such that we can determine their age structure. In general, the time until the internal node labeled $k$ is moved down by one layer is geometrically distributed with parameter $\frac{k(k+1)}{n^2}$, because $\frac{k(k+1)}{2}$ branches exist above this internal node. In the case of the root node, this expectation is $\frac{2}{n^2}$, as stated before.\\
The time until the current root node of $T_i$ vanishes under the $\textit{EMG}^\flat$ is therefore distributed as the sum of $n-1$ independent random variables $\rho_1,\dots,\rho_{n-1}$, where $\rho_k$ is geometrically distributed with parameter $\frac{k(k+1)}{n^2}$. Its expectation is the sum of the expectations of the $\rho_k$, i.e. $\sum_{k=1}^{n-1}\frac{n^2}{k(k+1)}=n^2\left(1-\frac{1}{n}\right)$. In the large-population limit, this corresponds to a rate of $2$.\\ 
The expected time until an internal node of $T_i$ becomes invisible, averaged over all nodes, is
\begin{eqnarray}
\frac{1}{n-1}\sum_{k=1}^n n^2 \sum_{j=k}^{n-1} \frac{1}{j(j+1)}&=&\frac{n^2}{n-1}\sum_{k=1}^{n-1}\left( \frac{1}{k}-\frac{1}{n-1}\right)\nonumber\\
&=&\frac{n^2}{n-1}a_{n-1}-\frac{n^2}{(n-1)^2}\nonumber\\
&\approx & n\log(n)\nonumber
\end{eqnarray}

Rescaling time, we obtain $\frac{2\log(n)}{n}\rightarrow 0$ as $n\rightarrow\infty$. We conclude that in a large evolving tree, most internal nodes (which correspond to coalescent events) only exist for a short time until they are removed by the dynamics.
\section{Discussion}
The \textit{Evolving Moran Genealogy} and time-reversed version reveal interesting properties of the genealogies generated by the neutral Moran process. We have used them to re-formulate classic diffusion-limit results on the \textit{MRCA} process, but also obtain exact expressions for the finite-population setting. Additionally, the distribution of \textit{MRCA} jumps during fixation periods becomes tractable in the $\textit{EMG}^\flat$, for both finite and limiting case. Of practical interest may be the fact that the $\textit{EMG}^\flat$ reduces the number of operations from $2$ to $1$ in contrast to the underlying Moran Model, if we think of the regrafting operation as one single operation.\\ 
It might prove insightful to extend this to other Moran-type population models, such as ones involving alleles with a selective advantage. In such models, the associated tree-valued processes need to be described and may not admit such simple definitions as those we find in the neutral case. Still, such constructions might enable a similar kind of analysis that we have performed here. In many settings, tree-valued dual processes have been described for the infinite-population limits, e.g. the ancestral selection graph \cite{krone:asg}, which involves both mutation and selection between and among types, and admits a graphical representation similar to the lookdown-construction under neutrality \cite{lenz:lookdownasg}, making many implicit features of the model accessible (see e.g. \cite{baake:mutseleceq}). Tree-valued constructions of finite time and population size may help here to bridge the gap between finite and infinite population case likewise.

\section*{Acknowledgments}
We would like to thank Anton Wakolbinger and Peter Pfaffelhuber for helpful suggestions, Jan Rolfes and Anna-Lena Tychsen for interesting discussions and two anonymous reviewers for their comments.\\
This work was financially supported by the German Research Foundation (DFG-SPP1590).\\

\bibliography{references}
\bibliographystyle{mystylefile}

\end{document}